\documentclass[submission,copyright,creativecommons]{eptcs}
\providecommand{\event}{SR 2016} % Name of the event you are submitting to
\usepackage{breakurl}             % Not needed if you use pdflatex only.

\usepackage{amsmath}
\usepackage{amssymb}
\usepackage{url}
\usepackage{amsthm}
\usepackage{graphicx}
\usepackage{amscd}

\usepackage{tikz} % added by SLR
\usepackage[boxed,linesnumbered,longend]{algorithm2e} % added by SLR
\usetikzlibrary{positioning}
\usetikzlibrary{shapes.geometric}

\usetikzlibrary{automata}
\usepackage{subcaption}

\newcommand{\name}[1]{\textsc{#1}}
\newcommand{\hide}[1]{}

\theoremstyle{definition}
\newtheorem{theorem}{Theorem}
\newtheorem{definition}[theorem]{Definition}

\newtheorem{corollary}[theorem]{Corollary}

\newtheorem{lemma}[theorem]{Lemma}
\newtheorem{observation}[theorem]{Observation}

\newtheorem{example}[theorem]{Example}
\newtheorem*{lemma*}{Lemma}
\newtheorem*{observation*}{Observation}

\title{Extending Finite Memory Determinacy to Multiplayer Games}
\author{St\'ephane Le Roux
\institute{D\'epartement d'Informatique\\ Universit\'e Libre de Bruxelles, Belgique}
\email{Stephane.Le.Roux@ulb.ac.be}
\and
Arno Pauly
\institute{ D\'epartement d'Informatique\\ Universit\'e Libre de Bruxelles, Belgium
\email{Arno.Pauly@cl.cam.ac.uk}}
}
\def\titlerunning{Extending Finite Memory Determinacy}
\def\authorrunning{S.~Le Roux \& A.~Pauly}
\begin{document}
\maketitle

\begin{abstract}
We show that under some general conditions the finite memory determinacy of a class of two-player win/lose games played on finite graphs implies the existence of a Nash equilibrium built from finite memory strategies for the corresponding class of multi-player multi-outcome games. This generalizes a previous result by  Brihaye, De Pril and Schewe. For most of our conditions we provide counterexamples showing that they cannot be dispensed with.

Our proofs are generally constructive, that is, provide upper bounds for the memory required, as well as algorithms to compute the relevant winning strategies.
\end{abstract}

\section{Introduction}
The usual model employed for synthesis are sequential two-player win/lose games played on finite graphs. The vertices of the graph correspond to states of a system, and the two players jointly generate an infinite path through the graph (the \emph{run}). One player, the protagonist, models the aspects of the system under the control of the designer. In particular, the protagonist will win the game iff the run satisfies the intended specification. The other player is assumed to be fully antagonistic, thus wins iff the protagonist loses. One then would like to find winning strategies of the protagonist, that is, a strategy for her to play the game in such a way that she will win regardless of the antagonist's moves. Particularly desirable winning strategies are those which can be executed by a finite automaton.

Classes of games are distinguished by the way the winning conditions (or more generally, preferences of the players) are specified. Typical examples include:
\begin{itemize}
\item Muller conditions, where only the set of vertices visited infinitely many times matters;
\item Parity conditions, where each vertex has a priority, and the winner is decided by the parity of the least priority visited infinitely many times;
\item Energy conditions, where each vertex has an energy delta (positive or negative), and the protagonist loses if the cumulative energy values ever drop below 0;
\item Discounted payoff conditions, where each vertex has a payoff value, and the outcome is determined by the discounted sum of all payoffs visited with some discount factor $0 < \lambda < 1$;
\item Combinations of these, such as energy parity games, where the protagonist has to simultaneously ensure that the least parity visited infinitely many times is odd and that the cumulative energy value is never negative.
\end{itemize}

Our goal is to dispose of two restrictions of this setting: First, we would like to consider any number of players; and second allow them to have far more complicated preferences than just preferring winning over losing. The former generalization is crucial in a distributed setting (also e.g.~\cite{depril2,bulling}): If different designers control different parts of the system, they may have different specifications they would like to enforce, which may be partially but not entirely overlapping. The latter seems desirable in a broad range of contexts. Indeed, rarely is the intention for the behaviour of a system formulated entirely in black and white: We prefer a programm just crashing to accidently erasing our hard-drive; we prefer a programm to complete its task in 1 minute to it taking 5 minutes, etc. We point to \cite{kupferman} for a recent survey on such notions of quality in synthesis.

Rather than achieving this goal by revisiting each individual type of game and proving the desired results directly (e.g.~by generalizing the original proofs of the existence of winning strategies), we shall provide a transfer theorem: In Theorem \ref{theo:maingraph}, we will show that (under some conditions), if the two-player win/lose version of a game is finite memory determined, the corresponding multi-player multi-outcome games all have finite memory Nash equilibria.

This result is more general than a similar one obtained by \name{Brihaye}, \name{De Pril} and \name{Schewe} \cite{depril2},\cite[Theorem 4.4.14]{depril}. A particular class of games covered by our result but not the previous one are (a variant of) energy parity games as introduced by \name{Chatterjee} and \name{Doyen} \cite{chatterjee4}. The high-level proof idea follows earlier work by the authors on equilibria in infinite sequential games, using Borel determinacy as a blackbox \cite{paulyleroux2}\footnote{Precursor ideas are also present in \cite{leroux3} and \cite{mertens} (the specific result in the latter was joint work with Neymann).} -- unlike the constructions there (cf.~\cite{paulyleroux3-cie}), the present ones however are constructive and thus give rise to algorithms computing the equilibria in the multi-player multi-outcome games given suitable winning strategies in the two-player win/lose versions.

Echoing \name{De Pril} in \cite{depril}, we would like to stress that our conditions apply to the preferences of each player individually. For example, some players could pursue energy parity conditions, whereas others have preferences based on Muller conditions: Our results apply just as they would do if all players had preferences of the same type.

For several of the conditions in our main theorem we also provide examples showing that they cannot be removed.

\section{Background}
\label{sec:background}

A two-player win/lose game played on a finite graph is specified by a directed graph $(V,E)$ where every vertex has an outgoing edge, a starting vertex $v_0\in V$, two sets  $V_1 \subseteq V$ and $V_2 := V \setminus V_1$, and a \emph{winning condition} $W\subseteq V^\omega$. Starting from  $v_0$, the players move a token along the graph, $\omega$ times, with player $a\in\{1,2\}$ picking and following an outgoing edge whenever the current vertex lies in $V_a$. Player $1$ wins iff the infinite sequence of visited vertices is in $W$.

For $a\in\{1,2\}$ let $\mathcal{H}_a$ be the set of finite paths in $(V,E)$ starting at $v_0$ and ending in $V_a$. Let $\mathcal{H} := \mathcal{H}_1 \cup \mathcal{H}_2$ be the possible \emph{finite histories} of the game, and let $[\mathcal{H}]$ be the infinite ones. For clarity we may write $[\mathcal{H}_g]$ instead of $[\mathcal{H}]$ for a game $g$. A \emph{strategy} of player $a\in\{1,2\}$ is a function of type $\mathcal{H}_a \to V$ such that $(v,s(hv)) \in E$ for all $hv\in \mathcal{H}_a$. A pair of strategies $(s_1,s_2)$ for the two players induces a run $\rho \in V^\omega$: Let $s := s_1 \cup s_2$ and set $\rho (0) := v_0$ and $\rho (n+1) := s(\rho(0)\rho(1)\ldots \rho(n))$. For all strategies $s_a$ of player $a$ let $\mathcal{H}(s_a)$ be the finite histories in $\mathcal{H}$ that are compatible with $s_a$, and let $[\mathcal{H}(s_a)]$ be the infinite ones. A strategy $s_a$ is said to be winning if $[\mathcal{H}(s_a)] \subseteq W$, \textit{i.e.} $a$ wins regardless of her opponent's moves.

A \emph{strategic implementation} for player $a$ using $m$ bits of memory is a function $\sigma: V \times \{0,1\}^m \to V\times \{0,1\}^m$ that describes the two simultaneous updates of player $a$ upon arrival at a vertex $v$ if its memory content was $M$ just before arrival: $(v,M)\mapsto \pi_2\circ\sigma(v,M)$ describes the memory update and $(v,M)\mapsto \pi_1\circ\sigma(v,M)$ the choice for the next vertex. This choice will be ultimately relevant only if $v\in V_a$, in which case we require that $(v,\pi_1\circ\sigma(v,M)) \in E$.

Together with an initial memory content $M_\epsilon\in \{0,1\}^m$, a strategic implementation provides a finite memory strategy. The memory content after some history is defined by induction: $M_{\sigma}(M_\epsilon,\epsilon):=M_\epsilon$ and $M_{\sigma}(M_\epsilon,hv) := \pi_2 \circ \sigma(v,M_{\sigma}(M_\epsilon,h))$ for all $hv\in \mathcal{H}$. The \emph{finite-memory strategy} $s_a$ induced by the strategic implementation $\sigma$ together with initial memory content $M_\epsilon$ is defined by $s_a(hv):= \pi_1\circ\sigma(v,M_{\sigma}(M_\epsilon,h))$ for all $hv\in \mathcal{H}_a$. If not stated otherwise, we will assume the initial memory to be $0^m$.

A (general) game played on a finite graph is specified by a directed graph $(V,E)$, a set of agents $A$, a cover $\{V_a\}_{a \in A}$ of $V$ via pairwise disjoint sets, the starting vertex $v_0$, and for each player $a$ a preference relation $\prec_a \subseteq [\mathcal{H}] \times [\mathcal{H}]$. The notions of strategies and induced runs generalize in the obvious way. In particular, instead of a pair of strategies (one per player), we consider families $(s_a)_{a \in A}$, which are called strategy profiles.

The concept of a winning strategy no longer applies though. Instead, we use the more general notion of a Nash equilibrium: A family of strategies $(s_a)_{a \in A}$ is a Nash equilibrium, if there is no player $a_0 \in A$ and alternate strategy $s'_{a_0}$ such that $a$ would prefer the run induced by $(s_a)_{a \in A \setminus \{a_0\}} \cup (s'_{a})_{a \in \{a_0\}}$ to the run induced by $(s_a)_{a \in A}$. Intuitively, no player can gain by unilaterally deviating from a Nash equilibrium. Note that the Nash equilibria in two-player win/lose games are precisely those pairs of strategy where one strategy is a winning strategy.

The transfer of results from the two-player win/lose case to the general case relies on the idea that each general game induces a collection of two-player win/lose games, namely the threshold games of the future games, as below.

\begin{definition}[Future game and one-vs-all threshold game] \hfill

Let $g = \langle (V,E), v_0, A, \{V_a\}_{a \in A}, (\prec_a)_{a \in A}\rangle$ be a game played on a finite graph.
\begin{itemize}
\item Let $a_0\in A$ and $\rho \in [\mathcal{H}]$, the one-vs-all threshold game $g_{a_0,\rho}$ for $a_0$ and $\rho$ is the win-lose two-player game played on $(V,E)$, starting at $v_0$, with vertex subsets  $V_{a_0}$ and $\bigcup_{a \in A \setminus \{a_0\}} V_a$, and with winning set $\{\rho'\in [\mathcal{H}]\,\mid\,\rho \prec_{a_0} \rho'\}$ for Player $1$.

\item Let $v\in V$. For paths $hv$ and $vh'$ in $(V,E)$ let $hv\hat{\,}vh' := hvh'$.

\item For all $h\in \mathcal{H}$ with last vertex $v$ let $g^{h} := \langle (V,E), v, A, \{V_a\}_{a \in A}, (\prec^{h}_a)_{a \in A}\rangle$ be called the future game of $g$ after $h$, where for all $\rho, \rho' \in [\mathcal{H}_{g^h}]$ we set $\rho \prec^{h}_a \rho'$ iff $h\hat{\,}\rho \prec_a h\hat{\,}\rho'$. If $s$ is a strategy in $g$, let $s^{h}$ be the strategy in $g^{h}$ such that $s^{h}(h') := s(h\hat{\,}h')$ for all $h'\in \mathcal{H}_{g^h}$.
\end{itemize}
\end{definition}

\begin{observation}\label{obs:strat-derived-game}
Let $g = \langle (V,E), v_0, A, \{V_a\}_{a \in A}, (\prec_a)_{a \in A}\rangle$ be a game played on a finite graph.
\begin{enumerate}
\item\label{obs:strat-derived-game1} $g$ and its thresholds games have the same strategies.
\item\label{obs:strat-derived-game2} for all $h,h'\in \mathcal{H}$ ending with the same vertex the games $g^{h}$ and $g^{h'}$ have the same (finite-memory)
strategies.
\item\label{obs:strat-derived-game3} $g$, its future games, and their thresholds games have the same strategic implementations.
\item\label{obs:strat-derived-game4} If a strategy $s_a$ in $g$ is finite-memory, for all $h\in \mathcal{H}$ the strategy $s_a^{h}$ is also finite-memory.
\end{enumerate}
\begin{proof}
We only prove the fourth claim. Since $s_a$ is a finite-memory strategy, it comes from some strategic implementation $\sigma$ with initial memory $M_\epsilon$. We argue that $\sigma$ with initial memory $M_{\sigma}(M_\epsilon,h))$ implements $s_a^{hv}$: First, $s_a^{hv}(v) = s_a(hv) = \pi_1\circ\sigma(v,M_{\sigma}(M_\epsilon,h)) = \pi_1\circ\sigma(v,M_{\sigma}(M_{\sigma}(M_\epsilon,h),\epsilon))$; second, for all $h'v'\in \mathcal{H}^{hv}$ we have $s_a^{hv}(vh'v') = s_a(hvh'v') = \pi_1\circ\sigma(v',M_{\sigma}(M_\epsilon,hvh')) = \pi_1\circ\sigma(v',M_{\sigma}(M_{\sigma}(M_\epsilon,h),vh'))$.
\end{proof}

\end{observation}

We will employ some additional restrictions on preferences: a preference relation $\prec \subseteq [\mathcal{H}] \times [\mathcal{H}]$ is called \emph{prefix-linear}, if $\rho \prec \rho' \Leftrightarrow h \hat{\,}\rho \prec h \hat{\,} \rho'$ for all $\rho , \rho', h \hat{\,}\rho \in [\mathcal{H}]$. It is \emph{prefix-independent}, if $\rho \prec \rho' \Leftrightarrow h  \hat{\,}\rho \prec \rho'$ and $\rho' \prec \rho \Leftrightarrow \rho' \prec h \hat{\,}\rho $ for all $\rho , \rho', h \hat{\,}\rho \in [\mathcal{H}]$. Clearly, a prefix-independent preference is prefix-linear.

As a further generalization, we will consider \emph{automatic-piecewise prefix-linear} preferences $\prec$. Here, there is an equivalence relation on $\mathcal{H}$ with equivalence classes (pieces for short) in $\overline{\mathcal{H}}$ and satisfying three constraints: First, the histories in the same piece end with the same vertex. Second, there exists a deterministic finite automaton, without accepting states, that reads histories and  such that two histories are equivalent iff reading them leads to the same states. Third, for all $h \hat{\,}\rho , h \hat{\,}\rho', h' \hat{\,}\rho,h' \hat{\,}\rho' \in [\mathcal{H}]$, if $ \overline{h'} = \overline{h}\in \overline{\mathcal{H}}$, then $h \hat{\,}\rho \prec h  \hat{\,}\rho' \Leftrightarrow h'  \hat{\,}\rho \prec h'  \hat{\,}\rho'$.

\begin{definition}
A preference relation $\prec$ is automatic-piecewise Mont if there is an equivalence relation on $\mathcal{H}$ that is decidable by a finite automaton (with the first two constraints as for automatic-piecewise prefix linearity) such that the following holds. For every run $h_0 \hat{\,}\rho\in [\mathcal{H}]$ that is regular (as a singleton language) and for every family $(h_n)_{n\in\mathbb{N}}$ of paths in $(V,E)$ such that $h_0 \hat{\,}h_1 \hat{\,}\dots  \hat{\,}h_n \in \overline{h_0}$ for all $n$, and such that $h_0 \hat{\,}\dots  \hat{\,}h_n \hat{\,}\rho \in [\mathcal{H}]$ for all $n$, if $h_0 \hat{\,}\dots  \hat{\,}h_n \hat{\,}\rho \prec h_0 \hat{\,}\dots  \hat{\,}h_{n+1} \hat{\,}\rho$ for all $n$ then $h_0\hat{\,}\rho \prec h_0 \hat{\,}h_1 \hat{\,}h_2 \hat{\,}h_3\dots$.

For a cycle $h$ starting and ending in some $v \in V$, let $h^{[0]} = v$ and $h^{[n+1]} = h^{[n]} \hat{\,} h$, and finally $h^{[\omega]} = \lim_{n \to \infty} h^{[n]}$. We call $\prec$ automatic-piecewise regular-Mont, if for any regular $h_0 \hat{\,} \rho \in [\mathcal{H}]$ and cycle $h$ in $(V,E)$, if $\forall n \in \mathbb{N} \quad h_0\hat{\,} h^{[n]} \in \overline{h_0}$ and $\forall n \in \mathbb{N} \quad h_0 \hat{\,} h^{[n]} \hat{\,} \rho \prec h_0 \hat{\,} h^{[n+1]} \hat{\,} \rho$, then $h_0 \hat{\,} \rho \prec h_0 \hat{\,} h^{[\omega]}$.
\end{definition}

\begin{definition}[Strict weak order]
 \label{def:strictweakorder}
Recall that a relation $\prec$ is called a \emph{strict weak order} if it satisfies:
\[\begin{array}{l@{\quad}l}
\forall x,\quad \neg(x\prec x)  \\
\forall x,y,z, \quad x\prec y \,\wedge\, y\prec z\,\Rightarrow\, x\prec z  \\
\forall x,y,z, \quad \neg(x\prec y) \,\wedge\, \neg(y\prec z)\,\Rightarrow\, \neg(x\prec z)
\end{array}\]
\end{definition}

Strict weak orders capture in particular the situation where each player cares only about a particular aspect of the run (e.g.~her associated personal payoff), and is indifferent between runs that coincide in this aspect but not others (e.g.~the runs with identical associated payoffs for her, but different payoffs for the other players).

\section{The transfer theorem}
\label{sec:main}

We will start this section with the statement of our first main result, showing how to transfer finite-memory determinacy from class of two-player win/lose games to the corresponding multi-player multi-outcome version. We informally sketch the proof. This is followed by the technical definitions and lemmata used in the formal proof, and then the proof itself. The section is completed by a discussion of the relevance and a comparison to prior results.

\begin{theorem}
\label{theo:maingraph}
Consider a game played by a set of players $A$ on a finite graph such that
\begin{enumerate}
\item\label{cond:the-graph-0} The $\prec_a$ are automatic-piecewise prefix-linear Mont strict weak orders with $k$ pieces.
\item\label{cond:the-graph-2} All one-vs-all threshold games of all future games are determined via strategies using $m$ bits of memory.
\end{enumerate}
Then the game has a Nash equilibrium in finite-memory strategies requiring $|A|(m + 2\log k) + 1$ bits of memory.
\end{theorem}

Definition~\ref{defn:ag-IAC} below rephrases Definitions 2.3 and 2.5 from~\cite{leroux3}: The guarantee of a player is the smallest set of runs that is upper-closed w.r.t. the strict-weak-order preference of the player and includes every incomparability class (of the preference) that contains any run compatible with a given strategy of the player in the subgame at any given finite history of the game. The best guarantee of a player consists of the intersection of all her guarantees over the set of strategies.

More plainly spoken, the best guarantee for a player at some history is the set of runs such that the player cannot unilaterally enforce something better (for him). We will then show that each player has indeed a strategy enforcing her guarantee. Note that the notion of best guarantee for a player does not at all depend on the preferences of the other players; and as such, it is rather strenuous to consider such runs to be optimal in some sense (cf.~Example \ref{example:guarantee}). However, we can construct a Nash equilibrium by starting with a strategy profile where everyone is realizing their guarantee, and then adding punishments against any deviators.

\begin{definition}[Player (best) future guarantee]\label{defn:ag-IAC}
Let $g$ be the game\\$\langle (V,E), v_0, A, \{V_a\}_{a \in A}, (\prec_a)_{a \in A}\rangle$ where $\prec_a$ is a strict weak order for some $a\in A$. For all $h\in \mathcal{H}$ and strategies $s_a$ for $a$ in $g^{h}$ let $\gamma_{a}(h,s_a) := \{\rho\in [\mathcal{H}_{g^h}] \,\mid\,\exists \rho'\in [\mathcal{H}_{g^h}(s_a)],\, \neg (\rho \prec_a^h \rho')\}$ be the player future guarantee by $s_a$ in $g^{h}$. Let $\Gamma_{a}(h):=\bigcap_{s_a}\gamma_a( h,s_a)$ be the best future guarantee of $a$ in $g^h$.
\end{definition}

\begin{example}
\label{example:guarantee}
Let the underlying graph be as in Figure~\ref{fig:fig1a}, where circle vertices are controlled by Player 1 and diamond vertices are controlled by Player 2. The preference relation of Player $1$ is $(ab)^\omega \succ_1 a(ba)^nx^\omega \succ_1 (ab)^ny^\omega$ and the preference relation of Player $2$ is $(ab)^\omega \succ_2 (ab)^ny^\omega \succ_2 a(ba)^{n}x^\omega$ (in particular, both players care only about the tail of the run).

Then $\Gamma_1(a) = \{(ab)^\omega\} \cup \{a(ba)^nx^\omega \mid n \in \mathbb{N}\}$ and $\Gamma_2(a) = [\mathcal{H}]$. Player $1$ realizing her guarantee means for her to move to $x$ immediately, thus forgoing any chance of realizing the run $(ab)^\omega$. The Nash equilibrium constructed in the proof of Theorem \ref{theo:maingraph} will be Player $1$ moving to $x$ and Player $2$ moving to $y$. Note that in this particular game, the preferences of Player $2$ have no impact at all on the Nash equilibrium that will be constructed.
\end{example}

\begin{figure}
\centering
\begin{subfigure}[b]{0.5\textwidth}
\begin{tikzpicture}[shorten >=1pt,node distance=1.6cm, auto]
  \node[state, initial] (x) {a};
  \node[draw, diamond] (y)  [right of = x] {b};
 \node[state] (x2) [below of = x] {x};
   \node[draw, diamond] (y2)  [below of = y] {y};

\path[->] (x) edge [bend left] node {} (y)
                edge node {} (x2)
                (x2) edge [loop left] node {} ()
                (y) edge [bend left] node {} (x)
                edge node {} (y2)
                (y2) edge [loop right] node {} ()       ;
 \end{tikzpicture}
  \caption{}
 \label{fig:fig1a}
\end{subfigure}
 \begin{subfigure}[b]{0.25\textwidth}
\begin{tikzpicture}[shorten >=1pt,node distance=1.6cm, auto]
  \node[state, initial] (x) {b};
  \node[state] (y)  [right of = x] {g};

\path[->] (x) edge [loop above] node {1} ()
                edge [bend left] node {0} (y)
                (y) edge [loop above] node {0} ()
                edge [bend left] node {0} (x);
 \end{tikzpicture}
  \caption{}
 \label{fig:finite-unbounded-memory}
   \end{subfigure}
 \caption{}
 \label{fig:misc}
 \end{figure}

\begin{lemma}\label{lem:guarantee-inclusion}
Let $g$ be a game on a graph, let $\prec_a$ be a strict weak order preference for some player $a$, let $h\in \mathcal{H}$, let $s_a$ be a strategy for $a$ in $g^{h}$, let $h' \in \mathcal{H}(s_a)$, and let $s'_a$ be a strategy for $a$ in $g^{h\hat{\,}h'}$.
\begin{enumerate}

\item\label{lem:guarantee-inclusion2} Then $h'\hat{\,}\gamma_a(h\hat{\,}h',s_a^{h'}) \subseteq \gamma_a(h,s_a)$ for all $h'\in \mathcal{H}(s_a)$.\\

\item\label{lem:guarantee-inclusion5} If $\gamma_a(h\hat{\,}h',s'_a) \subsetneq \gamma_a(h\hat{\,}h', s_a^{h'})$, there exists $\rho \in \gamma_a(h\hat{\,}h',s_a^{h'})$ such that $\rho \prec_a^{h\hat{\,}h'} \rho'$ for all $\rho' \in \gamma_a(h\hat{\,}h',s'_a)$.

\item\label{lem:guarantee-inclusion6} If $\gamma_a(h\hat{\,}h',s'_a) \subseteq \gamma_a(h\hat{\,}h', s_a^{h'})$ then $h'\hat{\,}\gamma_a(h\hat{\,}h',s'_a) \subseteq \gamma_a(h, s_a)$.

\end{enumerate}
\begin{proof}
\begin{enumerate}

\item Let $\rho \in \gamma_a(h\hat{\,}h',s_a^{h'})$, so by Definition~\ref{defn:ag-IAC} there exists $\rho' \in [\mathcal{H}(s_a^{h'})]$ such that $\neg(\rho \prec_a^{h\hat{\,}h'} \rho')$, \textit{i.e.} $\neg(h'\hat{\,}\rho \prec_a^{h} h'\hat{\,}\rho')$. So $h'\hat{\,}\rho \in \gamma_a(h,s_a)$ since $h'\hat{\,}\rho' \in [\mathcal{H}(s_a)] \subseteq \gamma_a(h,s_a)$.

\item Let $\rho \in  \gamma_a(h\hat{\,}h', s_a^{h'}) \setminus \gamma_a(h\hat{\,}h',s'_a)$, so $\rho \prec_a^{h\hat{\,}h'}\rho'$ for all $\rho' \in \gamma_a(h\hat{\,}h',s'_a)$ by Definition~\ref{defn:ag-IAC}.

\item Let $\rho \in \gamma_a(h\hat{\,}h',s'_a)$, so $h'\hat{\,}\rho \in h'\hat{\,}\gamma_a(h\hat{\,}h', s_a^{h'})$ by assumption, so $h'\hat{\,}\rho \in \gamma_a(h, s_a)$ by Lemma~\ref{lem:guarantee-inclusion}.\ref{lem:guarantee-inclusion2}.

\end{enumerate}
\end{proof}
\end{lemma}

\begin{lemma}
\label{lem:guarantee-ufm}
Let $g$ be a game on a finite graph with strict weak order $\prec_a$ for some $a\in A$, and let $m\in\mathbb{N}$. In each of the threshold games for $a$ of the future games let us assume that

\begin{enumerate}
\item\label{lem:guarantee-ufm1} if Player $1$ has a winning strategy, she has one with memory size $m$. Then for all $h \in \mathcal{H}$ there exists a strategy $s_a$ with memory size $m$ such that $\gamma_a( h, s_a) = \Gamma_a( h)$.

\item\label{lem:guarantee-ufm2} one of the players has a winning strategy with memory size $m$. Then $\Gamma_a( h)$ has a regular $\prec_a^h$-minimum for all $h \in \mathcal{H}$. %Taking this minimum as a threshold Player $2$ has a winning strategy for the threshold game for $a$ in $g^h$.
\end{enumerate}
\begin{proof}
\begin{enumerate}
\item For all $\rho \in [\mathcal{H}_{g^h}]$, we have $\rho \notin \Gamma_a(h)$ iff Player $1$ has a winning strategy in the threshold game for $a$ and $\rho$ in $g^h$. In this case let $s_a^p$ be a winning strategy with memory size $m$. For a game with $n$ vertices there are at most $(n2^{m})^{(n2^m)}$ strategy profiles using $m$ bits of memory (by the $\sigma$ representation), so the $s_a^p$ are finitely many, so at least one of them, which we name $s_a$, wins the threshold games for all $\rho \notin \Gamma_a(h)$. This shows that $\gamma_a( h, s_a) \subseteq \Gamma_a( h)$, so equality holds.

\item Towards a contradiction let us assume that $\Gamma_a(h)$ has a no $\prec_a^h$-least element, and let $\rho \in \Gamma_a(h)$, so there exists a $\prec_a^h$-smaller $\rho'\in \Gamma_a(h)$. Since $a$ has no winning strategy for the threshold game for $\rho'$ (of the future game at $h$), the determinacy assumption implies that the coalition of her opponents has one with memory size $m$. These strategies are finitely many, so one of them, $s_{-a}$, wins the threshold game for all $\rho \in \Gamma_a(h)$. So $\mathcal{H}_{g^h}(s_{-a}) \cap \Gamma_a(h) = \emptyset$, which contradicts Lemma~\ref{lem:guarantee-ufm}.\ref{lem:guarantee-ufm1}. Furthermore, the run induced by $s_{-a}$ and one finite-memory $s_a$ from Lemma~\ref{lem:guarantee-ufm}.\ref{lem:guarantee-ufm1} is regular.
\end{enumerate}
\end{proof}
\end{lemma}

\begin{lemma}\label{lem:ghf}
Let $g$ be a game on a finite graph with automatic piecewise prefix-linear strict weak order $\prec_a$ for some $a\in A$. If $h,h'\in H \in \overline{\mathcal{H}}$, then $\gamma_a(h,s_a) = \gamma_a(h',s_a)$ for all strategies $s_a$ for $a$ in $g^h$, and $\Gamma_{a}(h) = \Gamma_{a}(h')$.
\end{lemma}
\begin{proof}
By definition of the future games and automatic-piecewise prefix-linearity $\rho \prec_a^{h}\rho'$ iff $h\hat{\,}\rho \prec_a h\hat{\,}\rho'$ iff $h'\hat{\,}\rho \prec_a h'\hat{\,}\rho'$ iff $\rho \prec_a^{h'}\rho'$, so $\prec_a^{h} = \prec_a^{h'}$.
\end{proof}

\begin{lemma}\label{lem:opt-strat}
Let $g$ be a game on a finite graph with automatic-piecewise prefix-linear strict weak order $\prec_a$ for some $a\in A$. Let $m\in\mathbb{N}$ and assume that the threshold games for $a$ of the future games are determined \textit{via} size-$m$ strategies. Let $H \in \overline{\mathcal{H}}$.

\begin{enumerate}
\item\label{lem:opt-strat1} There exists a strategy $s_a^{H}$ with memory size $m$ such that $\gamma_a(h, s_a^{H}) = \Gamma_a(h)$ for all $h\in H$.

\item\label{lem:opt-strat2} There exists a size-$m$ strategy $s_{-a}^{H}$ for Player $2$ in $g^H$ (\textit{i.e.} $g^h$ for any $h\in H$) that is winning the threshold game for $a$ and $\rho$ in $g^H$ for all $\rho \in \Gamma_a(H)$ (\textit{i.e.} $\Gamma_a(h)$ for any $h\in H$).
\end{enumerate}
\begin{proof}
\begin{enumerate}
\item Lemma~\ref{lem:guarantee-ufm}.\ref{lem:guarantee-ufm1} provides a candidate, Lemma~\ref{lem:ghf} shows that it works.

\item Let $h\in H$ and let $\rho$ be one $\prec_a^h$-minimum of $\Gamma_a(h)$ by Lemma~\ref{lem:guarantee-ufm}.\ref{lem:guarantee-ufm2}. Since Player $1$ has no winning strategy in the threshold game for $a$ and $\rho$ in $g^h$, there exists a size-$m$ strategy $s_{-a}^{H}$ that makes Player $2$ win. By Lemma~\ref{lem:ghf} this strategy works also for $g^{h'}$ for all $h' \in H$.
\end{enumerate}
\end{proof}
\end{lemma}

Lemma~\ref{lem:max-guarant1-fm} below already uses all the assumptions used in Theorem~\ref{theo:maingraph}, but only for one given player.

\begin{lemma}
\label{lem:max-guarant1-fm}
Let $g$ be a game on a finite graph, and let some $\prec_a$ be an automatic-piecewise prefix-linear regular-Mont strict weak order with $k$ pieces.  Let $m\in\mathbb{N}$ and assume that the threshold games for $a$ of the future games are determined \textit{via} size-$m$ strategies. There is a strategy $s$ in $g$ such that  $\textrm{Reg} \cap \gamma_a( h, s^{h}) = \textrm{Reg} \cap \Gamma_a(h)$ for all $h \in \mathcal{H}$, and that uses $m + 2 \log k$ bits of memory, where $\textrm{Reg}$ denotes the set of all regular runs $\rho \in [\mathcal{H}]$.

If $\prec_a$ is as above, but even fulfills the Mont condition, then we can ensure $\gamma_a( h, s^{h}) = \Gamma_a(h)$ for all $h \in \mathcal{H}$.
\begin{proof}
We define a strategic implementation for $s$ in pseudocode in Algorithm~\ref{algo:g-strat}. The algorithm uses in particular that by Lemma~\ref{lem:opt-strat}.\ref{lem:opt-strat1} for any piece $\overline{h}$ there is a strategic implementation using $m$ bits for a strategy $s_a^{\overline{h}}$ such that $\gamma_a(h, s_a^{\overline{h}}) = \Gamma(h)$ for all $h\in \overline{h}$. An index of one of these strategic implementations is always stored, and the combined strategic implementation then follows the stored strategy as long as this one continues to realize the guarantee. If this is no longer the case, all of the $k$ strategic implementations are checked, and one is chosen that does realize the guarantee, and this one is followed from there onwards.

\begin{algorithm}
\SetKwProg{Fn}{Function}{ is}{end}
\SetKwFunction{KwStrategy}{Strategy}
\SetKwFunction{KwUpdatePiece}{updatePiece}
\SetKwFunction{KwUpdateLocal}{updateLocal}
\SetKwFunction{KwRealizeGua}{realizesGuarantee}
\SetKwData{Strat}{Strat}
\SetKwData{Mem}{Mem}
\SetKwData{Piece}{Piece}

\KwData{Current local strategy implementation \Strat\\ current local memory content \Mem\\ current piece \Piece}

\Fn{\KwUpdatePiece}{
\SetKwInOut{Input}{input}\SetKwInOut{Output}{output}
\Input{A piece $H$ of history and a vertex $v \in V$}

\Output{The piece of history after $H$ and then $v$}

}

\Fn{\KwUpdateLocal}{
\SetKwInOut{Input}{input}\SetKwInOut{Output}{output}
\Input{a strategy implementation s, a local memory content M, the vertex v the play is arriving in}
\Output{the updated local memory content M' and the vertex v' the strategy s wants to move to}
}

\Fn{\KwRealizeGua}{
\SetKwInOut{Input}{input}\SetKwInOut{Output}{output}
\Input{a strategy implementation s, a local memory content M, the current piece H}
\Output{a boolean answer whether (s, M) is realizing the guarantee at H}
}

\BlankLine

\Fn{\KwStrategy}{
\SetKwInOut{Input}{input}\SetKwInOut{Output}{output}
\Input{vertex v the play is arriving in}

\Output{vertex v' the strategy wants to move to}

\Piece := \KwUpdatePiece(\Piece, v)\;
(M', v') := \KwUpdateLocal(\Strat, \Mem, v)\;
\If{\KwRealizeGua(\Strat, M', \Piece)}{\Mem := M'\; \Return{v'}\;}
\Else{\ForEach{Strategy implementation s}{
(M',v') := \KwUpdateLocal(s, $0^m$, v)\;
\If{\KwRealizeGua(s, M', \Piece)}{
\Strat := s\;
\Mem := M'\;
\Return{v'}\;
}
}}
}
\caption{Strategy for player $a$ realizing guarantees}\label{algo:g-strat}
\end{algorithm}

{\bf Claim}: The strategy implemented by the Algorithm \ref{algo:g-strat} indeed satisfies the criteria.
\begin{proof}
Let $h_0 \in \mathcal{H}$. To show that $\gamma_a(h_0, s^{h_0}) = \Gamma_a(h_0)$ ($\textrm{Reg} \cap \gamma_a(h_0, s^{h_0}) = \textrm{Reg} \cap \Gamma_a(h_0)$), let $\rho \in \mathcal{H}(s^{h_0})$ ($\rho \in \textrm{Reg} \cap \mathcal{H}(s^{h_0})$) and let us make a case distinction. First case, $a$ changes strategies finitely many times along $\rho$. Let $h_1,\dots,h_n$ be such that for all $1 \leq i \leq n$ the $i$-th update along $\rho$ occurs at history $h'_i := h_0\hat{\,}h_1\hat{\,}\dots\hat{\,}h_i$. Applying Lemma~\ref{lem:guarantee-inclusion}.\ref{lem:guarantee-inclusion6} $n$ times yields $h_1\hat{\,}\dots \hat{\,}h_{n}\hat{\,}\gamma_a(h'_n,t_{h'_n}) \subseteq \dots\subseteq  h_{1}\hat{\,}\gamma_a(h'_1,t_{h'_1})\subseteq \gamma_a(h_0,t_{h_0}) \subseteq \Gamma_a(h_0)$. So $\rho \in \Gamma_a(h_0)$ since $\rho\in h_1\hat{\,}\dots\hat{\,}h_{n}\hat{\,}\gamma_a(h'_n,t_{h'_n})$.

Second case, $a$ changes strategies infinitely many times along $\rho$. Let $(h_n)_{n\geq 1}$ be the paths in $(V,E)$ such that the $n$-th change occurring strictly after $h_0$ occurs at history $h'_n :=h_0 \hat{\,} h_1\hat{\,}\dots \hat{\,}h_n$. By Lemmas~\ref{lem:guarantee-ufm}.\ref{lem:guarantee-ufm2} and \ref{lem:ghf}, for all $H\in \overline{\mathcal{H}}$ let $\rho_a^H$ be a regular $\prec_a^H$-minimum of $\Gamma_a(H)$. By Lemma~\ref{lem:guarantee-inclusion}.\ref{lem:guarantee-inclusion5}, for all $1\leq n$ there exists $\rho'\in \gamma_a(h'_{n+1}, t_{h'_n}^{h_{n+1}})$ such that $\rho' \prec_a^{h'_{n+1}} \rho''$ for all $\rho'' \in \gamma_a(h'_{n+1}, t_{h'_{n+1}})$, especially $\rho' \prec_a^{h'_{n+1}} \rho_a^{\overline{h'_{n+1}}}$. Since $h_{n+1}\rho' \in \gamma_a(h'_n,t_{h'_n})$ by Lemma~\ref{lem:guarantee-inclusion}.\ref{lem:guarantee-inclusion2}, we find $\rho_a^{\overline{h'_n}} \prec_a^{h'_n} h_{n+1} \hat{\,} \rho_a^{\overline{h'_{n+1}}}$. By finiteness of $\overline{\mathcal{H}}$ one $H\in \overline{\mathcal{H}}$ occurs infinitely many times as a $\overline{h'_n}$. For all $n \geq 1$ let $h'_{\varphi(n)}$ be the $n$-th corresponding history.

If $\rho$ is regular, there is some $f : \mathbb{N} \to \mathbb{N}$, finite path $h'$ and cycle $h$ such that $h_0 \hat{\,}h' \hat{\,} h^{[n]} = h'_{\varphi(f(n))}$ (and thus $h' \hat{\,} h^{[\omega]} = \rho$). The inequality above can then be rewritten $h_0 \hat{\,} h' \hat{\,} h^{[n]} \hat{\,} \rho_a^{\overline{h}} \prec_a h_0 \hat{\,} h' \hat{\,} h^{[n+1]} \hat{\,} \rho_a^{\overline{h}}$ for all $n\geq 1$, so $h_0 \hat{\,} h'\hat{\,}\rho_a^{\overline{h}} \prec_a h_0\hat{\,}\rho$ by the regular-Mont condition, so $h'\hat{\,}\rho_a^{\overline{h}} \prec_a^{h_0} \rho$.

In the general case, let $h = h'_{\varphi(1)}$, and let $h'$ be such that $h = h_0 \hat{\,} h'$. The inequality above can then be rewritten $h'_{\varphi(n)}\hat{\,}\rho_a^{\overline{h}} \prec_a h'_{\varphi(n+1)}\hat{\,}\rho_a^{\overline{h}}$ for all $n\geq 1$, so $h\hat{\,}\rho_a^{\overline{h}} \prec_a h_0\hat{\,}\rho$ by the Mont condition, so $h'\hat{\,}\rho_a^{\overline{h}} \prec_a^{h_0} \rho$.

Since $h'\hat{\,}\rho_a^{\overline{h}} \in \Gamma_a(h_0)$ by the finite case above, $\rho \in \Gamma_a(h_0)$ .
\end{proof}

Let us now analyze the pseudo-code in Algorithm~\ref{algo:g-strat} and find out how much memory suffices. The algorithm keeps track of the current piece $H$ of history so far, which by assumption requires $\log k$ bits. It also keeps track of an index of the current strategic implementation, which again requires $\log k$ bits (as we need at most one strategic implementation per piece). Finally, we use $m$ bits  for the memory content $M$.

 \end{proof}
\end{lemma}

\begin{proof}[Proof of Theorem \ref{theo:maingraph}]
We combine Lemmas~\ref{lem:guarantee-ufm}.\ref{lem:guarantee-ufm1} and \ref{lem:max-guarant1-fm} to obtain a finite memory strategy $s_a$ for each player $a$ such that $\textrm{Reg} \cap \gamma_a( h, s_a^{h}) = \textrm{Reg} \cap \Gamma_a(h)$ for all $h \in \mathcal{H}$. The run $\rho_{NE}$ induced by this strategy profile will be the run induced by the Nash equilibrium. As $\rho_{NE}$ is induced by a finite-memory strategy profile, $\rho_{NE} \in \textrm{Reg}$, and in particular, for any decomposition $\rho_{NE} = h \hat{\,} \rho$ and any player we find that $\rho \in \Gamma_a(h)$. Now we need to ensure that no one has any incentive to deviate.

For all $a\in A$ and all pieces $H \in \overline{\mathcal{H}}$ of histories ending in $V_a$ let $s_{-a}^{H}$ be the strategy from Lemma~\ref{lem:opt-strat}.\ref{lem:opt-strat2}. For all $a\in A$ let $a$ play as follows: Keep playing according to $\rho_{NE}$ until one player $b$ deviates from $\rho_{NE}$ at some history $h_D$. If $b = a$, \textit{i.e.}, the last vertex of $h_D$ is in $V_a$, let $a$ do whatever; otherwise let $a$ play according to $s_{-b}^{\overline{h_D}}$ in $g^{h_D}$, \textit{i.e.} let $a$ take part in the one-vs-all coalition that ensures that $b$ cannot obtain anything $\prec_a^{h_D}$-better than $\rho_{\overline{h_D}}$ in $g^{h_D}$, so $\neg(\rho_{\overline{h_D}} \prec_b^{h_D}\rho_D)$, where $\rho_D$ is the new run in $g^{h_D}$ after deviation by $b$. Let $h_D\hat{\,}\rho = \rho_{NE}$. By construction of $\rho_{NE}$ and Lemma~\ref{lem:max-guarant1-fm}, $\rho \in \Gamma_b(h_D)$, so $\neg(\rho\prec_b^{h_D}  \rho_{\overline{h_D}} )$. So $\neg(\rho_{NE} \prec_b h_D\hat{\,}\rho_D)$ since $\prec_b$ is a strict weak order, \textit{i.e.} $b$ has no incentive to perform a deviation.

Each player $a$ needs to run the intended strategies $s_b$ for all other players $b$, too, in order to be able to detect deviation. This requires $|A|(m + 2\log k)$ bits by Lemma~\ref{lem:max-guarant1-fm}. If a deviation is detected, the punishment strategy for that history is executed instead. (Note that the history where the deviation happened includes the information on who deviated; and that the players already keep track of the history in the strategy of Lemma~\ref{lem:max-guarant1-fm}.) The punishment phase requires $\log k$ memory to record which strategy to follow, and $m$ bits to follow it. However, the original memory can be repurposed, by using just one extra bit to determine whether deviation ever occurred. This yields the overall memory bound $|A|(m + 2\log k) + 1$.
\end{proof}

\section{Discussion}\vspace{-1em}
\subsubsection*{Comparison to previous work}
As mentioned above, a similar but weaker result (compared to our Theorem \ref{theo:maingraph}) has previously been obtained by \name{Brihaye}, \name{De Pril} and \name{Schewe} \cite{depril2},\cite[Theorem 4.4.14]{depril}. They use cost functions rather than preference relations. Our setting of strict weak orders is strictly more general \footnote{For example, the lexicographic combination of two payoff functions can typically not be modeled as a payoff function, as $\mathbb{R} \times \{0,1\}$ (with lexicographic order) does not embed into $\mathbb{R}$ as a linear order.}\label{page:footnote}. However, even if both frameworks are available, it is more convenient for us to have results formulated via preference relations rather than cost functions: Cost functions can be translated immediately into preferences, whereas translating preferences to cost functions is more cumbersome. In particular, it can be unclear to what extend \emph{nice} preferences translate into \emph{nice} cost functions. Note also that prefix-linearity for strict weak orders is more general than prefix-linearity for cost functions.

As a second substantial difference, \cite[Theorem 4.4.14]{depril} requires either prefix-independent cost functions and finite memory determinacy of the induced games, or prefix-linear cost functions and positional subgame-perfect strategies. Their subgame-perfection assumption essentially means that they assume the result of Lemma \ref{lem:max-guarant1-fm}. In particular,\cite[Theorem 4.4.14]{depril} cannot be applied to energy parity games, where finite prefixes of the run do impact the overall value for the players, and where at least the protagonist requires memory to execute a winning strategy.

The Mont condition is absent in \cite{depril2}, but it can be shown that their other requirements imply a slight weakening of the Mont condition (which still suffices for our proof).

Before \cite{depril,depril2}, it had already been stated by \name{Paul} and \name{Simon} \cite{soumya} that multi-player multi-outcome Muller games have Nash equilibria consisting of finite memory strategies. As (two-player win/lose) Muller games are finite memory determined \cite{gurevich2}, and the corresponding preferences are obviously prefix independent, this result is also a consequence of \cite[Theorem 4.4.14]{depril}. Another result subsumed by \cite[Theorem 4.4.14]{depril} (and subsequently our main theorem) is found in \cite{depril3} by \name{Brihaye}, \name{Bruy\`ere} and \name{De Pril}.
\vspace{-1em}
\subsubsection*{Algorithmic considerations}
Let us briefly consider the algorithmic price to pay for the extension for the two-player win/lose case to the multi-player multi-outcome situation. Let us assume that for some class of games satisfying the criteria of Theorem \ref{theo:maingraph}, computing a winning strategy of a two-player win/lose game of size $n$ has winning strategies takes $f(n)$ time. Let us further assume that, given finite memory strategies of size up to $m$ bits for each player, we can decide in time $g(n,m)$ who is winning. Additionally, let us assume that a multi-player multi-outcome game of size $n$ induces up to $t(n)$ one-vs-all threshold games of size $n$.

We can find for each $h \in \mathcal{H}$ some strategy $s_a$ with memory size $m$ such that $\gamma_a( h, s_a) = \Gamma_a( h)$ by brute force: We are investigating up to $t(n)$ induced one-vs-all threshold games, and are asking for a winning strategy in each of them, which could take up to $t(n)(f(n) + g(n,m))$ time. In any concrete example, though, a much more efficient construction is to be expected.

In Lemma \ref{lem:max-guarant1-fm}, we need use the result above for each combination of memory content ($2^{m}$ different values) and piece of the partition ($k$). Thus, we are spending up to $2^mkt(n)(f(n) + g(n,m))$ time to obtain sufficiently many strategies to realize the guarantee everywhere, which we can combine in linear time to yield the single strategy output strategy.

In the proof of Theorem \ref{theo:maingraph}, we invoke Lemma \ref{lem:max-guarant1-fm} once per player, which costs $O(|A|2^mkt(n)(f(n) + g(n,m))$ time. In addition, we need $k$ many winning strategies in an induced one-vs-all threshold game, for additional cost of $O(kf(n))$ time -- in total, we are still at $O(|A|2^mkt(n)(f(n) + g(n,m))$.

Let us introduce some realistic but simplifying additional assumptions. We would expect $t$ to be of the form $2^{O(n)}$. The parameter $k$ will be dominated by $2^{O(n)}$ in most situations. Typical finite-memory determined win/lose games might required exponential memory (if measured by number of states), but as we measure $m$ in bits, also $2^m$ would be absorbed into $2^{O(n)}$. By dropping the distinctions between finding the winning strategy and determining who is winning, and noting that $|A| \leq n$, we arrive at an overall complexity of $2^{\mathcal{O}(n)}f(n)$. The additional factor of $2^n$ will in some cases worsen the asymptotic runtime significantly (e.g.~for parity games, subexponential algorithms are known \cite{paterson}), but in others, pales against the complexity for solving the two-player win/lose case.
\vspace{-1em}
\subsubsection*{On uniform finite memory determinacy}
The requirement in Theorem \ref{theo:maingraph} that there is a uniform memory bound sufficient for all threshold games is not dispensible. Mean-payoff parity games, for example, satisfy all other criteria, yet lack finite memory Nash equilibria, as the following example shows.

Let $g$ be the one-player game in Figure~\ref{fig:finite-unbounded-memory}. The payoff of a run that visits the vertex $g$ infinitely often is the limit (inferior or superior) of the average payoff. It is zero if $g$ is visited finitely many times only. For any threshold $t \in \mathbb{R}$, if $t < 1$, the player has a winning finite-memory strategy: cycle $p$ times in $b$, where $p > \frac{1}{1-t}$, visit $g$ once, cycle $p$ times in $b$, and so on. If $t \geq 1$, the player has no winning strategy at all. So the thresholds games of $g$, and likewise for the future game of $g$, are finite-memory determined. The game has no finite memory Nash equilibrium nonetheless, since the player can get a payoff as closed to $1$ as she wants, but not $1$.

Note that the preceding example could also be used for discounted-payoff parity games. However, these also fail the automatic-piece prefix-linearity condition.

Uniform finite memory determinacy can sometimes be recovered by considering $\varepsilon$-versions instead: We partition the payoffs in blocks of size $\varepsilon$, and let the player be indifferent within the same block. Clearly any Nash equilibrium from the $\varepsilon$-discretized version yields an $\varepsilon$-Nash equilibrium of the original game. If the original preferences were prefix-independent, the modified preferences still are. Moreover, as there are now only finitely many relevant threshold games per graph, their uniform finite memory determinacy follows from mere finite memory determinacy. In such a situation, our result allows us to conclude that multi-player multi-outcome games have finite memory $\varepsilon$-Nash equilibria.
\vspace{-1em}
\subsubsection*{On the Mont condition}
We can exhibit a prototypic example for how failure of the regular-Mont condition translates into the absence of Nash equilibria:

\begin{example}[\footnote{This example is based on an example communicated to the authors by Axel Haddad and Thomas Brihaye, which in turn is based on a construction in \cite{haddad}.}]
\label{example:mont}
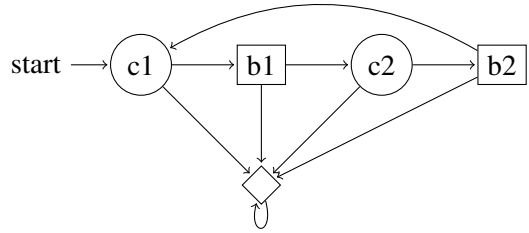
\begin{figure}
\centering
\begin{tikzpicture}[shorten >=1pt,node distance=1.6cm, auto]
  \node[draw, circle, initial] (x) {c1};
  \node[draw] (y)  [right of = x] {b1};
  \node[draw, circle] (x2)  [right of = y] {c2};
  \node[draw] (y2)  [right of = x2] {b2};
  \node[draw,diamond] (z) [below of = y] {};

\path[->] (x) edge node {} (y)
                (x) edge node {} (z)
                (y) edge node {} (x2)
                edge node {} (z)
                (x2) edge node {} (y2)
                edge node {} (z)
                (y2) edge [bend right] node {} (x)
                edge node {} (z)
                (z) edge [loop below] node {} ();
 \end{tikzpicture}
  \caption{The graph for the game in Example \ref{example:mont}}
 \label{fig:Mont-necessary}
 \end{figure}

The game $g$ in Figure~\ref{fig:Mont-necessary} involves Player $1$ ($2$) who owns the circle (box) vertices. Who owns the diamond is irrelevant. The payoff for Player $1$ ($2$) is the number of visits to a box (circle) vertex, if this number is finite, and is $-1$ otherwise. Let $s_1$ be the positional strategy where Player $1$ chooses $b1$ when in $c1$ and the diamond when in $c2$, and let $t_2$ be the positional strategy where Player $2$ always chooses the diamond. With $s_1$ Player $1$ secures payoff $1$, and with $t_2$ Player $2$ makes sure that Player $1$ is not getting more than that. Let $s_2$ be any positional strategy for Player $2$, and let $t_1$ be the positional strategy where Player $1$ always choses the diamond.  With $s_2$ Player $1$ secures payoff $1$, and with $t_1$ Player $1$ makes sure that Player $2$ is not getting more than that. Therefore the threshold games of $g$ are positionally determined, and likewise for the future games of $g$. The game $g$ has no Nash equilibrium nonetheless: in the run induced by a putative NE, one of the players has to choose the diamond at some point (to avoid payoff $-1$), but by postponing this choice to next time, the player can increase her payoff by $1$. This shows the relevance of the Mont condition.
\end{example}
\vspace{-1em}

\subsection*{On automatic-piecewise prefix-linearity}
While the definition of automatic-piecewise prefix-linearity may seem a bit complicated at first glance, the notion seems to fit in well with finite memory strategies: Intuitively, the requirement is merely that however we split a run into a finite prefix and an infinite tail, the contribution of the prefix to the value for the player factors via something expressible by a fixed (i.e.~independent of the length of the prefix) number of bits, and does so via a finite automaton. Whenever this is not satisfied, one would expect that in principle a finite memory strategy may fail (compared to an unrestricted strategy), simply because it cannot properly account for the contribution of the finite prefix it has seen so far.

Of the popular winning conditions, many are actually prefix-independent, such as parity, Muller, mean-payoff, cost-Parity, cost-Streett \cite{fijalkow2}, etc. Clearly, any combination of prefix-independent conditions itself will be prefix-independent. Typical examples of non-prefix independent conditions are reachability, energy, and discounted payoff. We can easily verify that combining a reachability or energy condition with any prefix-linear condition yields an automatic-piecewise prefix-linear condition (provided that energy is bounded).

Discounted payoff can be problematic, though. Here each vertex is assigned a payoff $a_v$, a discount factor $\delta \in (0,1)$ is chosen, and the value of a run $\rho = v_0v_1\ldots$ is $\sum_{i = 0}^\infty v_i\delta^i$. While discounted payoff on its own is of course prefix-linear, it does not combine well with other criteria: For example, in a generalized discounted payoff parity game the question whether we prefer a tail with a better discounted payoff but worse least priority to a tail with worse discounted payoff but better least priority may depend on the precise value of the payoff obtained in the history so far, as well as the length of the history (as later contributions to payoff count less). This is too much information for a finite automaton to remember, thus, generalized discounted payoff parity games do not satisfy the criterion for being automatic-piecewise prefix-linear.

\subsubsection*{Applications}
We shall briefly mention two classes of games covered by our main theorem, but not by the results from \cite{depril,depril2}, (bounded) energy parity games and a variant of reachability + mean-payoff games. For more details, we refer to \cite{paulyleroux4-arxiv}.

Energy games were first introduced in \cite{chakrabarti}: Two players take turns moving a token through a graph, while keeping track of the \emph{current energy level}, which will be some integer. Each move either adds or subtracts to the energy level, and if the energy level ever reaches $0$, the protagonist loses. These conditions were later combined with parity winning conditions in \cite{chatterjee4} to yield energy parity games as a model for a system specification that keeps track of gaining and spending of some resource, while simultaneously conforming to a parity specification.

In both \cite{chakrabarti} and \cite{chatterjee4} the energy levels are a priori not bounded from above. This is a problem for the applicability of Theorem \ref{theo:maingraph}, as unbounded energy parity preferences are not automatic-piecewise prefix-linear in the general case. In \cite{bouyer}, two versions of bounded energy conditions were investigated: Either any energy gained in excess of the upper bound is just lost (as in e.g.~recharging a battery), or gaining energy in excess of the bound leads to a loss of the protagonist (as in e.g.~refilling a fuel tank without automatic spill-over prevention). We are only concerned with the former. As a finite automaton can easily keep track of the energy level between $0$ and the upper bound, bounded energy parity preferences are automatic-piecewise prefix linear. We can also show that these games are uniformly finite-memory determined, and with some extra work obtain (see \cite{paulyleroux4-arxiv}):

\begin{corollary}
All multiplayer multioutcome energy parity games have Nash equilibria in finite memory strategies. Let $A$ be the set of players, $n$ the size of the graph and $W$ the largest energy delta. Further let $E$ be the maximum difference between the energy maximum and the energy minimum for some player. Then $1 + |A|nE^{|A|}\log (2nW) + (|A|^2 + |A|)\log nE$ bits of memory suffice.
\end{corollary}

In our second example, each player has both a mean-payoff goal and a reachability objective. Maximizing the mean-payoff, however, takes precedence, and the reachability objective only becomes relevant as a tie-breaker. These conditions are not expressible via some payoff or cost function\footnote{Compare the footnote on Page \ref{page:footnote}.}, but are still automatic-piecewise prefix-independent strict weak orders. As the two-player win/lose games can easily be reduced to mean-payoff games, we obtain uniform finite-memory determinacy, and thus the existence of finite-memory Nash equilibria by Theorem \ref{theo:maingraph}.

We leave the investigation whether winning conditions defined via $\mathrm{LTL}[\mathcal{F}]$ or $\mathrm{LTL}[\mathcal{D}]$ formulae \cite{kupferman2} match the criteria of Theorem \ref{theo:maingraph} to future work. Another area of prospective examples to explore are multi-dimensional objectives as studied e.g.~in \cite{raskin,raskin4}.

\bibliographystyle{eptcs}
\bibliography{../../../spieltheorie}

\end{document}